\documentclass[aps,prd,twocolumn,nofootinbib,superscriptaddress]{revtex4-2}

\usepackage[T1]{fontenc}
\usepackage[utf8]{inputenc}
\usepackage{graphicx}
\usepackage{amsmath,amssymb,amsthm}
\usepackage{bm}

\usepackage{orcidlink}
\usepackage{hyperref}

\newtheorem{theorem}{Theorem}
\newtheorem{definition}{Definition}

\begin{document}

\title{Causal Rigidity and the Single-Unit Universe: Integrating the Alexandrov-Zeeman and Unruh Clock Scales}

\author{Karl Svozil\,\orcidlink{0000-0001-6554-2802}}
\email{karl.svozil@tuwien.ac.at}
\homepage{http://tph.tuwien.ac.at/~svozil}

\affiliation{Institute for Theoretical Physics,
TU Wien,
Wiedner Hauptstrasse 8-10/136,
1040 Vienna,  Austria}

\date{\today}

\begin{abstract}
We unify two complementary viewpoints on relativistic spacetime and the counting of fundamental constants. Operationally, Matsas, Pleitez, Saa, and Vanzella (MPSV) have recently argued that relativistic spacetime requires only a single fundamental dimensional constant. Mathematically, theorems due to Alexandrov and Zeeman demonstrate that the light-cone structure determines the spacetime geometry only up to a conformal factor. We show that these approaches are mutually reinforcing: the Alexandrov-Zeeman theorems establish the rigid conformal structure of spacetime, while the ``bona fide clock'' required by MPSV serves the necessary mathematical role of breaking the dilation symmetry. We provide a formal derivation proving that the normalization of a single clock worldline is sufficient to select a unique metric from the conformal class, thereby clarifying that the number of fundamental constants is exactly one.
\end{abstract}

\maketitle

\section{Introduction}

The question ``How many fundamental dimensional constants does physics require?'' has been a subject of enduring debate, most famously in the ``trialogue'' between Duff, Okun, and Veneziano (DOV)~\cite{DuffOkunVeneziano}. While there is consensus that only dimensionless quantities are ultimately observable, disagreements persist regarding the number of independent \emph{standards} (or base units) necessary to describe dimensionful observables.

Recently, Matsas, Pleitez, Saa, and Vanzella (MPSV)~\cite{matsas-2024} revisited this issue from a \emph{spacetime-based} perspective. They argue that once a relativistic spacetime is assumed as the starting point, the units required to construct that spacetime are sufficient to express all observables. Specifically, they conclude that in a relativistic spacetime, the number of fundamental dimensional constants is \emph{one}.

Parallel to this metrological debate, the mathematical relativity community has long established---via theorems by Alexandrov~\cite{alex3} and Zeeman~\cite{Zeeman}---that the null (light-cone) structure of spacetime is incredibly rigid. These theorems state that any map preserving the causal order is a Lorentz transformation up to a global dilation (scale factor).

In this paper, we unify these two perspectives. In Sec.~\ref{sec:Operational}, we review the operational argument provided by MPSV. In Sec.~\ref{sec:Geometric}, we review the geometric constraints imposed by causal automorphisms. Finally, in Sec.~\ref{sec:Formal}, we provide a formal derivation demonstrating that the ``bona fide clock'' of MPSV is the precise mathematical object required to break the conformal symmetry left open by incidence geometry.

\section{The Operational Perspective: Measuring Space with Time}
\label{sec:Operational}

MPSV argue that to construct a relativistic spacetime physically, one requires specific apparatuses. Unlike Galilean spacetime, which requires both rulers and clocks (because space and time are decoupled~\cite{matsas-2024}), relativistic spacetime requires only ``bona fide clocks.''

To prove that spatial rulers are redundant, MPSV utilize a protocol attributed to William George Unruh. Consider three inertial clocks $C_1, C_2, C_3$:
\begin{enumerate}
    \item $C_1$ travels from an observer $A$ to $B$.
    \item $C_2$ travels from $B$ back to $A$.
    \item $C_3$ remains at rest with $A$.
\end{enumerate}
If $\tau_1$ and $\tau_2$ are the proper time durations measured by the traveling clocks, and $\tau_3$ is the duration measured by the stationary clock,
the spatial distance $D$ between $A$ and $B$ is given by the height corresponding to the base $\tau_3$ of the triangle formed by sides $\tau_1, \tau_2, \tau_3$ (derived from Heron's formula~\cite{matsas-2024}):
\begin{equation}
D =  \frac{\sqrt{(\tau_3^2 - \tau_1^2 - \tau_2^2)^2 - 4\tau_1^2\tau_2^2}}{2\tau_3}.
\label{eq:Unruh}
\end{equation}
This formula yields a spatial length in units of time (seconds). It demonstrates physically that the spatial metric can be derived entirely from time measurements, rendering the meter a ``disposable'' unit and $c$ a mere conversion factor.
When the two traveling clocks record the same proper time, $\tau_1 = \tau_2 \equiv \tau$, the expression simplifies to
$
D = (\tau_3/2) \sqrt{1 - \left(2\tau/\tau_3\right)^2}
$.

\section{The Geometric Perspective: Alexandrov-Zeeman Theorems}
\label{sec:Geometric}

We now turn to the mathematical structure of Minkowski spacetime.
We consider the manifold $M \simeq \mathbb{R}^n$ (with $n \geq 3$), equipped with the standard flat metric $\eta_{ab} = \text{diag}(-1, 1, \dots, 1)$. Throughout this paper, we adopt the ``mostly plus'' signature convention $(-,+,+,\dots,+)$, so that timelike vectors have negative squared norm and spacelike vectors have positive squared norm.
The defining feature of relativistic geometry is the causal structure, encoded either by the incidence of light cones or the partial causal ordering of events.

A fundamental question in mathematical relativity is the extent to which this causal structure alone determines the geometry. This is resolved by a class of rigidity results known as Alexandrov-Zeeman type theorems. We state these results with their precise hypotheses, as the conditions are essential to their validity.

\subsection{Alexandrov: The Rigidity of Light Cones}

In the 1950s, A.~D.~Alexandrov proved that the geometry of spacetime is encoded in the ``skeleton'' of light rays~\cite{alex3}.

\begin{theorem}[Alexandrov]
\label{theorem:Alexandrov}
Let $f: M \to M$ be a \emph{bijection} of the entire spacetime manifold $M \simeq \mathbb{R}^n$ with $n \ge 3$. If $f$ maps null lines to null lines (i.e., preserves the light-cone incidence structure), then $f$ is a linear conformal automorphism.
\end{theorem}

The requirement that $f$ be a \emph{global} bijection of the entire spacetime is crucial. Local maps preserving null structure need not extend to conformal transformations; indeed, general relativity admits local diffeomorphisms that preserve the causal structure within a region but are not globally conformal. The theorem's power derives precisely from its global character.

\subsection{Zeeman: Causality Implies the Lorentz Group}

In 1964, E.~C.~Zeeman strengthened this perspective by focusing on the partial ordering of events~\cite{Zeeman}. We define the causal relations as follows: $x \prec y$ (read ``$x$ chronologically precedes $y$'') if a future-directed timelike or null signal can travel from $x$ to $y$, and $x \succ y$ if $y \prec x$.

\begin{theorem}[Zeeman]
Let $f: M \to M$ be a bijection of the entire Minkowski spacetime $M \simeq \mathbb{R}^n$ with $n \ge 3$. Suppose $f$ preserves the \emph{full} causal order, meaning:
\begin{equation}
x \prec y \iff f(x) \prec f(y) \quad \text{and} \quad x \succ y \iff f(x) \succ f(y).
\end{equation}
Then $f$ is a composition of:
\begin{enumerate}
    \item an orthochronous Lorentz transformation $\Lambda \in O^\uparrow(1, n-1)$,
    \item a spacetime translation $a \in \mathbb{R}^n$, and
    \item a global dilation $x \mapsto \lambda x$ with $\lambda > 0$.
\end{enumerate}
\end{theorem}

Several remarks on the hypotheses are in order:
\begin{itemize}
    \item \textit{Dimensionality:} The restriction $n \ge 3$ is essential. In $(1+1)$-dimensional spacetime, the theorem fails: the null lines form two independent families, and maps can independently rescale each family while preserving causal order.

    \item \textit{Bijectivity and global domain:} The map must be defined on and surject onto the \emph{entire} Minkowski space. This excludes local causal isomorphisms, which are far more numerous and less constrained.

    \item \textit{Full causal structure:} Preserving only the future-directed relation $\prec$ is insufficient. One must preserve both temporal orientations; equivalently, the map must preserve the relation of being causally connectable (whether to the past or future). Maps preserving only one direction could, in principle, include time-reversing transformations not generated by the orthochronous Lorentz group.
\end{itemize}

\subsection{Limitations and Extensions}

It is important to recognize the scope of these theorems. They apply to flat Minkowski spacetime as a global entity. Extensions to curved spacetimes require substantial additional machinery: In general relativity, the celebrated result of Malament~\cite{Malament-1977} establishes that the causal structure of a distinguishing spacetime determines the metric up to a conformal factor, but the proof strategy and required conditions differ significantly from the flat case. For the present discussion, we restrict attention to special relativity, where the Alexandrov-Zeeman framework applies directly.

\subsection{The Conformal Ambiguity}

Both theorems lead to the same physical conclusion: the causal structure determines the metric $g_{ab}$ up to a \emph{conformal factor}. Under the strict conditions of mapping the whole of Minkowski space to itself bijectively, this factor reduces to a global positive constant $\lambda$:
\begin{equation}
    g'_{ab} = \lambda^2 \eta_{ab}.
\end{equation}

This leads to a profound physical consequence: A universe defined solely by light rays (or causal order) cannot distinguish a second from a millennium. The transformation $x^\mu \to \lambda x^\mu$ for any $\lambda > 0$ is a valid symmetry of the causal structure. The ``shape'' of physics---the pattern of which events can influence which---is preserved, but all durations and lengths are uniformly rescaled. To distinguish absolute time scales, one must introduce additional structure that breaks this dilation symmetry.

\section{Formal Unification: The Clock as Symmetry Breaker}
\label{sec:Formal}

We now formalize the connection between the MPSV clock and the Alexandrov-Zeeman conformal factor. We show that specifying the normalization of a single clock worldline is sufficient to fix the global scale $\lambda$.

\subsection{Definitions}

\begin{definition}[Conformal Class]
Let $\eta$ be the standard Minkowski metric. The \emph{conformal class} $[\eta]$ determined by the causal structure is the set of metrics $\{\lambda^2 \eta : \lambda > 0\}$.
\end{definition}

\begin{definition}[Proper Time]
Let $\gamma: I \to M$ be a smooth timelike curve parameterized by an arbitrary parameter $s \in I \subseteq \mathbb{R}$. The \emph{proper time} along $\gamma$ with respect to a metric $g$ is:
\begin{equation}
\tau_g(s) = \int_{s_0}^{s} \sqrt{-g(\dot{\gamma}(s'), \dot{\gamma}(s'))} \, ds',
\label{eq:proper-time-integral}
\end{equation}
where $s_0 \in I$ is a reference parameter value (corresponding to the initial event) and dot denotes $d/ds'$.
\end{definition}

\begin{definition}[Bona Fide Clock]
A \emph{bona fide clock} consists of:
\begin{enumerate}
    \item a timelike worldline $\Gamma \subset M$, considered as a point set;
    \item a physical mechanism that produces readings $\tau_{\mathrm{clock}}: \Gamma \to \mathbb{R}$ along the worldline.
\end{enumerate}
The clock is \emph{ideal} with respect to the physical metric $g_\star$ if its readings coincide with the proper time:
\begin{equation}
\tau_{\mathrm{clock}}(\gamma(s)) = \tau_{g_\star}(s) + \mathrm{const},
\label{eq:ideal-clock-condition}
\end{equation}
for any parameterization $\gamma: I \to \Gamma$.
\end{definition}

The key distinction is that the worldline $\Gamma$ and the reading function $\tau_{\mathrm{clock}}$ are \emph{physically given} by the apparatus, independent of any coordinate choice or metric assumption. The condition~\eqref{eq:ideal-clock-condition} then becomes a constraint that the unknown metric $g_\star$ must satisfy.

\subsection{Uniqueness Theorem}

\begin{theorem}[Single Clock Fixes Scale]
Let $M \cong \mathbb{R}^4$ be equipped with a causal order determining the conformal class $[\eta]$. Let $(\Gamma, \tau_{\mathrm{clock}})$ be a single ideal clock following an inertial trajectory. Then the physical metric $g_\star \in [\eta]$ is uniquely determined.
\end{theorem}

\begin{proof}
Since the causal structure fixes the metric up to a global scale (via the Alexandrov-Zeeman theorems for $n \ge 3$), the physical metric must take the form:
\begin{equation}
g_\star = \lambda^2 \eta, \quad \lambda > 0,
\end{equation}
where $\eta$ is an arbitrary reference metric in standard Minkowski coordinates.

Let $\gamma: \mathbb{R} \to \Gamma$ be \emph{any} smooth parameterization of the worldline by a parameter $s$. The clock's physical mechanism provides the reading $\tau_{\mathrm{clock}}$ along the worldline; differentiating with respect to $s$ yields the empirically determined tick-rate $d\tau_{\mathrm{clock}}/ds$. For an inertial clock, this rate is constant.

Differentiating the ideal clock condition~\eqref{eq:ideal-clock-condition} with respect to $s$ yields:
\begin{equation}
\frac{d\tau_{\mathrm{clock}}}{ds} = \sqrt{-g_\star(\dot{\gamma}, \dot{\gamma})} = \sqrt{-\lambda^2 \eta(\dot{\gamma}, \dot{\gamma})} = \lambda \sqrt{-\eta(\dot{\gamma}, \dot{\gamma})}.
\end{equation}

Define $\alpha(s) \equiv \sqrt{-\eta(\dot{\gamma}(s), \dot{\gamma}(s))}$, which is the norm of the tangent vector in the reference metric. This quantity depends on the arbitrary parameterization $s$, but the ratio
\begin{equation}
\lambda = \frac{d\tau_{\mathrm{clock}}/ds}{\alpha(s)}
\end{equation}
is \emph{independent} of the choice of parameterization. To see this, consider a reparameterization $s \to \tilde{s}(s)$. Then $d\tau_{\mathrm{clock}}/d\tilde{s} = (d\tau_{\mathrm{clock}}/ds)(ds/d\tilde{s})$ and $\tilde{\alpha} = \alpha \cdot |ds/d\tilde{s}|$, so their ratio is invariant.

Since $d\tau_{\mathrm{clock}}/ds$ is physically given by the clock mechanism and $\alpha(s)$ is computed from the reference coordinates, the scale factor $\lambda$ is uniquely determined. Thus $g_\star = \lambda^2 \eta$ is unique.
\end{proof}

\subsection{The Hierarchy of Metrics}

It is instructive to explicitly distinguish the roles of the metric tensors discussed above:

\begin{enumerate}
    \item \textbf{The Reference Metric ($\eta_{ab}$):} An arbitrary mathematical auxiliary (the standard Minkowski metric) that correctly describes the causal incidence (light cones) but possesses no physical scale.

    \item \textbf{The Conformal Class ($[\eta]$):} The set of all possible metrics allowed by the Alexandrov-Zeeman theorems. This is a one-parameter family:
    \begin{equation}
        g'_{ab} = \lambda^2 \eta_{ab}, \quad \forall \lambda \in (0, \infty).
    \end{equation}
    At this level, the theory cannot distinguish between a second and a millennium; any value of $\lambda$ preserves the causal order.

    \item \textbf{The Physical Metric ($g_\star$):} The unique element of the conformal class selected by the physical apparatus. It is obtained by evaluating at the specific scale factor determined by the clock:
    \begin{equation}
        g_\star = \lambda_{\mathrm{phys}}^2 \, \eta_{ab}.
    \end{equation}
\end{enumerate}

Thus, while the causal structure generates the conformal equivalence class, the single bona fide clock collapses this class to the unique physical metric.

\section{Discussion}

The formal derivation above confirms the intuition that the counting of fundamental constants in Special Relativity is exactly one. The single unit is the physical realization of the mathematical mechanism required to break the dilation symmetry of incidence geometry.

\begin{itemize}
    \item \textit{Incidence Relations (Alexandrov/Zeeman):} These fix the conformal class $[\eta]$. Physically, they define the ``shape'' of spacetime and the structural role of $c$, but leave the scale undetermined.
    \item \textit{The Clock:} This provides the constraint that the metric must reproduce the clock's physical readings as proper time, thereby eliminating the dilation sector.
\end{itemize}

Therefore, the symmetry group of the underlying structure is the Poincar{\'e} group extended by dilations. The causal structure provides the Poincar{\'e} sector, and the single unit (the clock) fixes the scale.

Intuitively, the mechanism by which a single clock fixes the conformal factor can be understood through three complementary lenses. Geometrically, the causal structure (as rigidly reflected in Alexandrov--Zeeman--type results) is analogous to a map of a coastline that perfectly captures the shape of bays and peninsulas---preserving all angles---but lacks a scale bar. In this conformal picture, one cannot distinguish between a map of a tiny island and a massive continent; zooming in or out (global dilation) changes nothing about the geometric shape. The introduction of a bona fide clock acts effectively as drawing a fixed legend on this map. By declaring a specific physical interval of clock time to have a definite value (say, one second), the ``zoom level'' of the entire spacetime manifold is instantly frozen, converting purely relative shape into absolute geometry.

Algebraically, this distinction arises from the disparate nature of null versus timelike vectors. Light rays travel along null geodesics where the spacetime interval vanishes ($ds^2 = 0$). Because zero is invariant under multiplicative scaling ($0 \cdot \lambda^2 = 0$), light is inherently ``blind'' to the conformal factor; it perceives the causal connectivity but not the metric scale. A clock, however, follows a timelike trajectory whose tangent can be normalized to a non-zero constant (e.g., $ds^2 = -1$ per unit of proper time). Unlike the null case, this normalization is sensitive to scaling: $-1 \cdot \lambda^2 \neq -1$ unless $\lambda = 1$. Thus, the mere specification of a trajectory with fixed non-zero squared interval per tick mathematically forbids the arbitrary rescaling allowed by the light cones and singles out a unique metric from the conformal class.

Physically, this symmetry breaking is rooted in the nature of matter. Pure conformal symmetry characterizes idealized massless fields, which lack an intrinsic length or time scale. However, the construction of a physical clock---whether an atom with a characteristic transition frequency or a macroscopic oscillator---inevitably involves massive matter and thus introduces a dimensionful parameter. In quantum theory, mass is associated with a characteristic length scale via the Compton wavelength, and bound systems have characteristic periods. Once we adopt one such period as a unit of time, we effectively inject a fundamental ruler into the geometry. The clock is not merely a passive observer but a physical agent that operationally breaks the global scale invariance of the causal structure and thereby fixes the metric scale.

It is important to distinguish the strict mathematical necessity of the clock from the operational demonstration provided by MPSV. As shown in our Theorem\ref{theorem:Alexandrov}, the mere \emph{existence} of a bona fide clock---defined by the normalization condition $g(\dot{\gamma}, \dot{\gamma})=-1$---is sufficient to fix the conformal factor $\lambda$. Mathematically, no further protocol is required; the geometry is uniquely determined. The value of the MPSV protocol (Heron's formula) lies in demonstrating the \emph{metrological sufficiency} of this time standard. It answers the practical question: ``Granted that the geometry is fixed, can we measure spatial intervals without reintroducing a ruler?'' By deriving spatial distance explicitly from time measurements, MPSV confirms that the single constant provided by the clock serves as the measure for the entire spacetime manifold, not just the temporal axis.

Finally, regarding the role of $c$: kinematically, MPSV~\cite{matsas-2024} correctly identify $c$ as a mere unit converter between the fundamental time and the derivative space. Yet, dynamically, $c$ retains a privileged status. It is the unique velocity parameter that ensures the form invariance of the equations of motion, such as those of electrodynamics~\cite{svozil-relrel,svozil-2001-convention}. Thus, while the clock fixes the scale of the universe, $c$ defines the rigid phenomenal structure necessary for the clock to operationally tick in such a way as to guarantee the form invariance of the physical laws under Lorentz boosts.

\begin{acknowledgments}
This text was partially created and revised with assistance from one or more of the following large language models:  Gemini 3 Pro
and gpt-5.1-high. All content, ideas, and prompts were provided by the author.

This research was funded in whole or in part by the Austrian Science Fund (FWF) Grant \href{https://doi.org/10.55776/PIN5424624}{10.55776/PIN5424624}.
The author acknowledges TU Wien Bibliothek for financial support through its Open Access Funding Programme.

The author declares no conflict of interest.
\end{acknowledgments}

\bibliography{svozil}

\end{document}